\documentclass{article}

\usepackage{amsmath, amsthm, amssymb}
\usepackage{graphicx}
\usepackage{epstopdf}
\usepackage{amssymb}
\usepackage{url}
\usepackage{mathrsfs}

	\normalbaselines						  %

\newcommand{\R}{{\mathbb  R}}

\newcommand{\Z}{{\mathbb  Z}}

\newcommand{\N}{{\mathbb  N}}

\newcommand{\C}{{\mathbb  C}}

\newcommand{\fdot}{\,\cdot\,}

\makeatletter
\def\Ddots{\mathinner{\mkern1mu\raise\p@
\vbox{\kern7\p@\hbox{.}}\mkern2mu
\raise4\p@\hbox{.}\mkern2mu\raise7\p@\hbox{.}\mkern1mu}}
\makeatother

\newcommand{\cH}{\mathcal{H}}

\newcommand{\f}{\varphi}

\newcommand{\p}{\mathbb{P}}

\DeclareMathOperator{\spa}{span}

\DeclareMathOperator{\clos}{clos}

\DeclareMathOperator{\dist}{dist}

\newcommand{\ci}[1]{_{ {}_{\scriptstyle #1}}}
\newcommand{\ti}[1]{_{\scriptstyle \text{\rm #1}}}

\catcode`@=11
\chardef\mathlig@atcode\count255

\def\actively#1#2{\begingroup\uccode`\~=`#2\relax\uppercase{\endgroup#1~}}
\def\mathlig@gobble{\afterassignment\mathlig@next@cmd\let\mathlig@next= }
\def\mathlig@delim{\mathlig@delim}
\def\mathlig@defcs#1{\expandafter\def\csname#1\endcsname}
\def\mathlig@let@cs#1#2{\expandafter\let\expandafter#1\csname#2\endcsname}
\def\mathlig@appendcs#1#2{\expandafter\edef\csname#1\endcsname{\csname#1\endcsname#2}}

\def\mathlig#1#2{\mathlig@checklig#1\mathlig@end\mathlig@defcs{mathlig@back@#1}{#2}\ignorespaces}

\def\mathlig@checklig#1#2\mathlig@end{%
 \expandafter\ifx\csname mathlig@forw@#1\endcsname\relax
 \expandafter\mathchardef\csname mathlig@back@#1\endcsname=\mathcode`#1%
 \mathcode`#1"8000\actively\def#1{\csname mathlig@look@#1\endcsname}%
 \mathlig@dolig#1\mathlig@delim
\fi
\mathlig@checksuffix#1#2\mathlig@end
}

\def\mathlig@checksuffix#1#2\mathlig@end{%
\ifx\mathlig@delim#2\mathlig@delim\relax\else\mathlig@checksuffix@{#1}#2\mathlig@end\fi
}
\def\mathlig@checksuffix@#1#2#3\mathlig@end{%
\expandafter\ifx\csname mathlig@forw@#1#2\endcsname\relax\mathlig@dosuffix{#1}{#2}\fi
\mathlig@checksuffix{#1#2}#3\mathlig@end
}

\def\mathlig@dosuffix#1#2{%
\mathlig@appendcs{mathlig@toks@#1}{#2}%
\mathlig@dolig{#1}{#2}\mathlig@delim
}

\def\mathlig@dolig#1#2\mathlig@delim{%
 \mathlig@defcs{mathlig@look@#1#2}{%
 \mathlig@let@cs\mathlig@next{mathlig@forw@#1#2}\futurelet\mathlig@next@tok\mathlig@next}%
 \mathlig@defcs{mathlig@forw@#1#2}{%
  \mathlig@let@cs\mathlig@next{mathlig@back@#1#2}%
  \mathlig@let@cs\checker{mathlig@chck@#1#2}%
  \mathlig@let@cs\mathligtoks{mathlig@toks@#1#2}%
  \expandafter\ifx\expandafter\mathlig@delim\mathligtoks\mathlig@delim\relax\else
  \expandafter\checker\mathligtoks\mathlig@delim\fi
  \mathlig@next
 }%
 \mathlig@defcs{mathlig@toks@#1#2}{}%
 \mathlig@defcs{mathlig@chck@#1#2}##1##2\mathlig@delim{%
  \ifx\mathlig@next@tok##1%
   \mathlig@let@cs\mathlig@next@cmd{mathlig@look@#1#2##1}\let\mathlig@next\mathlig@gobble
  \fi
  \ifx\mathlig@delim##2\mathlig@delim\relax\else
   \csname mathlig@chck@#1#2\endcsname##2\mathlig@delim
  \fi
 }
  \ifx\mathlig@delim#2\mathlig@delim\else
  \mathlig@defcs{mathlig@back@#1#2}{\csname mathlig@back@#1\endcsname #2}%
 \fi
}%

\catcode`@\mathlig@atcode

\mathchardef\ordinarycolon\mathcode`\:
\def\vcentcolon{\mathrel{\mathop\ordinarycolon}}
\mathlig{:=}{\vcentcolon=}
\mathlig{::=}{\vcentcolon\vcentcolon=}

\newtheorem{theo}{Theorem}[section]
\newtheorem{cor}[theo]{Corollary}

\newtheorem{prop}[theo]{Proposition}
\newtheorem{mr}[theo]{Conjecture}

\newtheorem{rem}[theo]{Remark}

\begin{document}

\title{A new numerical approach to Anderson (de)localization}
\author{Constanze~Liaw\\Department of Mathematics\\Texas A\&M University\\Mailstop 3368\\College Station, TX  77843\\USA\\email: conni@math.tamu.edu}
\thanks{The author is  partially supported by the NSF grant DMS-1101477.}


\begin{abstract}
We develop a new approach for the Anderson localization problem. The implementation of this method yields strong numerical evidence leading to a (surprising to many) conjecture: The two dimensional discrete random Schr\"odinger operator with small disorder allows states that are dynamically delocalized with positive probability. This approach is based on a recent result by Abakumov--Liaw--Poltoratski which is rooted in the study of spectral behavior under rank-one perturbations, and states that every non-zero vector is almost surely cyclic for the singular part of the operator.

The numerical work presented is rather simplistic compared to other numerical approaches in the field. Further, this method eliminates effects due to boundary conditions.

While we carried out the numerical experiment almost exclusively in the case of the two dimensional discrete random Schr\"odinger operator, we include the setup for the general class of Anderson models called Anderson-type Hamiltonians.

We track the location of the energy when a wave packet initially located at the origin is evolved according to the discrete random Schr\"odinger operator.

This method does not provide new insight on the energy regimes for which diffusion occurs.
\end{abstract}

\maketitle

\section{Introduction}

In 1958 P.~W.~Anderson \cite{And1958} suggested that sufficiently large impurities in a semi-conductor could lead to spatial localization of electrons, called Anderson localization.
Although many physicists consider the problem solved, many mathematical questions with striking physical relevance remain open. The field has grown into a rich mathematical theory (see \cite{Germ} and \cite{Klo2} for the study of different Anderson models; for refined notions of Anderson localization see \cite{SIMSULE} and \cite{ExSpec}).

We consider the discrete random Schr\"odinger operator in dimension $d$ which is given by the self-adjoint operator
$$
H_\omega = - \bigtriangleup + \sum_{i\in \Z^d} \omega_i <\fdot, \delta_i> \delta_i
\text{ on }l^2(\Z^d).
$$
Here $\delta_i\in l^2(\Z^d)$ assumes the value 1 in the $i-$th entry, $i = (i_1, i_2, \hdots, i_d)$, and zero in all other entries (see equation \eqref{e-example} below for an example in two dimensions). The random variables $\omega_i$ are i.i.d.~with uniform distribution in $[-c,c]$, i.e.~according to the probability distribution $\mathbb{P} = (2c)^{-1} \Pi_i  \chi\ci{[-c,c]} dx$. The Laplacian describes a crystal with atoms located at the integer lattice points $\Z^d$. Adding the random part can be interpreted as having the atoms being not perfectly on the lattice points, but randomly displaced.

This paper pertains to one of the ``weaker" definitions of localization which is equivalent to the property
that the spectrum of the operator is almost surely purely singular.

It is known that for $d=1$, Anderson localization is produced by random disorders of any strength and at all energies (analytic proof, see e.g.~\cite{CFKS}, \cite{CL}, or \cite{FP}).

For $d\ge 2$, localization is proved analytically for disorders $c$ above a certain dimension-dependent threshold $c_d$. The first result of this type can be found in \cite{FS}. Simpler proofs and better constants can be found in \cite{AizMol1993} as well as \cite{SIMREV}.

Diffusion may hence only occur for small disorder $c$. This is precisely the question we are addressing:
Does the two dimensional discrete random Schr\"odinger operator exhibit Anderson localization for small disorder with non-zero probability?

The main contribution of this paper is the introduction of a new numerical approach, and its implementation for the two dimensional discrete random Schr\"odinger operator. This application supports the following conjecture which overthrows the widely spread belief that localization takes place for random disorders of any strength.

\begin{mr}[Delocalization conjecture]\label{t-mr}
For disorder $c \lesssim 0.7$, the two dimensional discrete random Schr\"odinger operator does not exhibit Anderson localization with positive probability; in the sense that it has non-zero absolutely continuous spectrum with positive probability. In particular, we do not have dynamical localization with positive probability for small disorder.
\end{mr}

While this conjecture is based on deep analytical results, we included an explicit statement of the main tool, see Corollary \ref{c-tool}, so that the applied numerical sections \ref{s-setup} through \ref{s-final} are essentially self-contained. We would like to point out the important feature of this Corollary which makes a numerical experiment feasible: It suffices to track the evolution (under the random Hamiltonian) of just one vector!

In Section \ref{s-PRE}, we introduce Anderson-type Hamiltonians, a very general notion of Anderson model which includes many of those studied in literature. The methods described within this paper can be extended to most Anderson-type Hamiltonians. We further explain the notion of singular and absolutely continuous spectrum, as well as the corresponding parts of the operator.
In Section \ref{s-theoretical} we state an improvement, Theorem \ref{t-AndHam}, of a result by Jaksic and Last concerning the cyclicity of vectors for the Anderson-type Hamiltonian. We further explain how analytic results are used to prove the main tool, Corollary \ref{c-tool}.
Section \ref{s-setup} is devoted to a description of the numerical experiment.
A summary of numerical results and the conclusions can be found in Section \ref{ss-results}.
In Section \ref{s-supp} we verify the performance of the method in many examples, e.g.~for large disorder, for the free/unperturbed two dimensional discrete Schr\"odinger operator, for the one dimensional discrete random Schr\"odinger operator. We further include an investigation of the distribution of energies after repeated application of the random operator of a wave packet initially located at the origin.
We briefly remark on computing and memory requirements in Section \ref{s-final}.

{\bf Acknowledgements.} The author would like to thank A.~Poltoratski for suggesting the initial mathematical idea of the experiment, and G.~Berkolaiko for the insightful discussions concerning many aspects of this research as well as for reading and making useful comments on most parts of this paper. 
Further, she would like to thank J.~Kuehl for running initial experiments using a code of his, as well as for being such a wonderful husband.

\section{Preliminaries}\label{s-PRE}

\subsection{Anderson-type Hamiltonians, discrete Schr\"odinger operator}\label{ss-Hw}

While the numerical experiment within pertains to the discrete random Schr\"odinger operator, we define so-called Anderson-type Hamiltonians which were first introduced in \cite{JakLast2000}. The advantage of making this general definition is that this notion is a generalization of many Anderson models discussed in literature. In particular, the method described within this paper can be applied to many other interesting Anderson models.

For $n\in\N$ consider the probability space $\Omega_n=(\R, \mathcal B, \mu_n)$, where $\mathcal B$ is the Borel sigma-algebra on $\R$ and $\mu_n$
is a Borel probability measure. Let $\Omega=\prod_{n=0}^\infty \Omega_n$ be a product space with the probability measure $\p$ on
$\Omega$  introduced as the product measure of the corresponding measures on $\Omega_n$ on the product sigma-algebra $\mathcal A$. The elements of $\Omega$ are points in
$\R^\infty$,  $\omega=(\omega_1,\omega_2,...)$ for $ \omega_n\in\Omega_n$.

Let $\cH$ be a separable Hilbert space and let $\{\f_n\}_{n\in\N}$ be a countable
collection of unit vectors in $\cH$. For each $\omega\in\Omega$ define an Anderson-type Hamiltonian on $\cH$ as a self-adjoint operator formally given by
\begin{equation}\label{Model}
H_\omega = H + V_\omega, \text{ where } V_\omega = \sum\limits_n \omega_n <\fdot, \f_n>\f_n.
\end{equation}

Except for degenerate cases, the perturbation $V_\omega$ is almost surely a non-compact operator. It is hence not possible to apply results from classical perturbation theory to study the spectra of $H_\omega$, see e.g.~\cite{birst} and \cite{katobook}.

In the case of an orthogonal sequence $\{\f_n\}$, this operator was studied in \cite{JakLast2000} and \cite{JakLast2006}.
Probably the most important special case of an Anderson-type Hamiltonian is the discrete random Schr\"odinger operator on $l^2(\Z^d)$
\begin{align}\label{d-RandSchr}
Hf(x)=-\bigtriangleup f (x) = - \sum\limits_{|n|=1} (f(x+n)-f(x)), \text{ with } \f_n(x)=\delta_n(x)=
\left\{\begin{array}{ll}1&x=n\in\Z^d,\\ 0&\text{else.}\end{array}\right.
\end{align}

\subsection{Singular and absolutely continuous parts of normal operators}\label{pre-normal}
Recall that an operator in a separable Hilbert space is called normal if $T^*T= TT^*$. By the spectral theorem operator $T$ is unitarily equivalent to $M_z$, multiplication by the independent variable $z$, in a direct sum of  Hilbert spaces
$$\cH = \oplus \int \cH(z)\, d\mu(z)$$ where $\mu$ is a scalar positive  measure  on $\C$. The measure $\mu$ is called a scalar spectral measure
of $T$.

If $T$ is a unitary or self-adjoint operator, its spectral measure $\mu$ is supported on the unit circle or on the real line, respectively.
Via Radon decomposition, $\mu$ can be decomposed into a singular and absolutely continuous parts $\mu=\mu\ti{s}+\mu\ti{ac}$.
The singular component $\mu\ti{s}$ can be further split into singular continuous and pure point parts.
For unitary or self-adjoint $T$ we denote by $T\ti{ac}$  the restriction of $T$ to its absolutely continuous part, i.e.~$T\ti{ac}$ is unitarily equivalent to $M_t\big|_{\oplus \int \cH(t) d\mu\ti{ac}(t)}.$ Similarly, define the singular, singular continuous and the pure point parts of  $T$, denoted by $T\ti{s}$, $T\ti{sc}$ and $T\ti{pp}$, respectively.

\section{Theoretical background and the main tool}\label{s-theoretical}
Let us explain how the main theoretical tool used to indicate delocalization, Corollary \ref{c-tool}, is deduced from the following theorem. Let us remind the reader that we use the term delocalization to mean the existence of absolutely continuous spectrum almost surely, and that such delocalization implies dynamical delocalization.

A sequence $\{\f_n\}\subset \cH$ is called a \emph{representing system}, if every vector $\f\in\cH$ can be represented as a series
$
\f=\sum a_n \f_n
$
that converges with respect to the norm of $\cH$.
Note that, bases are representing systems. However, unlike in the case of a basis, a representation of a vector need not be unique.

We use the following Theorem, see \cite{AbaLiawPolt}.

\begin{theo}\label{t-AndHam}
Let $H_\omega$ be the Anderson-type Hamiltonian introduced in equation \eqref{Model}.
Suppose that the probability measure $\p$ is a product of absolutely continuous measures and  $\{\f_n\}$ is a representing system in $\cH$. Assume that
there exists a vector $\psi\in\cH$ that is cyclic for  $H_\omega$, $\p$-almost surely. Then any non-zero $\f\in \cH$ is cyclic for $H_\omega$, $\p$-almost surely.
\end{theo}

It is well-known that if an Anderson-type Hamiltonian $H_\omega$  is purely singular almost surely then it is cyclic almost surely. Equivalently, if such an operator is not cyclic with positive probability, then there are energies which are diffusive with non-zero probability. A proof of almost-sure cyclicity of the singular part $(H_\omega)\ti{s}$ and almost-sure cyclicity of certain specific vectors can be found in \cite{JakLast2006} and for the discrete Schr\"odinger operator in \cite{Sim1994}. Together with the latter theorem and the fact that the $\delta_n$ form a basis of $l^2(\Z^d)$ we obtain the following result.

The following two statements were formulated by A.~Poltoratski (private communications).

\begin{cor}
Assume the hypotheses of Theorem \ref{t-AndHam}.
If Anderson localization occurs, then the orbit of any non-zero vector under the operator $H_\omega$ is almost surely dense in the Hilbert space $\mathcal{H}$.

In other words, fix $0\neq f\in \cH$.
If $H_\omega$ has purely singular spectrum almost surely, then for all $v\in\cH$ with norm $ \|v\|_\cH = 1$ the distance from the vector $v$ to the span of the orbit of $f$ under the operator is zero with respect to $\p$-almost surely, i.e.~with respect to $\p$-almost surely
$$
\dist(v,\spa\{H^k_\omega f: k\in \{0, 1, 2, \hdots, n\}\})\stackrel{n\to\infty}{\longrightarrow} 0 .
$$
\end{cor}

More specifically, we will apply the following immediate consequence.

\begin{cor}\label{c-tool}
Consider the discrete random Schr\"odinger operator given by equations \eqref{Model} and \eqref{d-RandSchr} in dimension $d=2$. Let $\omega_i$, $i\in \Z^2$, be i.i.d.~random variables with uniform (Lebesgue) distribution on $[-c,c]$, $c>0$. To prove delocalization (i.e.~the existence of absolutely continuous spectrum with positive probability), it suffices to find $c>0$ for which the distance
\begin{align}\label{e-dist}
\dist(\delta_{11}, \spa\{H_\omega^k \delta_{00}: k = 0, 1, 2, \hdots, n\}) \nrightarrow 0 \text{ as }n\to\infty
\end{align}
with positive probability.
\end{cor}

\begin{rem} The converse of Corollary \ref{c-tool} is not true. Hence we cannot draw any conclusions, if the distance between a fixed (unit) vector and the subspace generated by the orbit of another vector tends to zero. In particular, we cannot conclude that there must be localization. Even if we show \eqref{e-dist} for many or `all' vectors (instead of just $\delta_{11}$), it could be possible that the absolutely continuous part has multiplicity one and that $\delta_{00}$ is cyclic, that is, $l^2(\Z^2) = \clos\spa\{H_\omega^k \delta_{00}: k \in \N\cup\{0\}\}$.
\end{rem}

\section{Method of numerical experiment}\label{s-setup}

Let us explain the computational approach used to indicate diffusion.

Consider the discrete Schr\"odinger operator given by \eqref{Model} and \eqref{d-RandSchr} with random variable $\omega$ distributed according to the hypotheses of Corollary \ref{c-tool}. Fix the vectors $\delta_{00}\in l^2(\Z^2)$ and $\delta_{11}\in l^2(\Z^2)$, that is
\begin{align}\label{e-example}
\delta_{00}=
\begin{pmatrix}
\ddots&\vdots&\vdots&\vdots&\vdots&\vdots&\Ddots\\
\hdots&0&0&0&0&0&\hdots\\
\hdots&0&0&0&0&0&\hdots\\
\hdots&0&0&1&0&0&\hdots\\
\hdots&0&0&0&0&0&\hdots\\
\hdots&0&0&0&0&0&\hdots\\
\Ddots&\vdots&\vdots&\vdots&\vdots&\vdots&\ddots
\end{pmatrix} ,
\delta_{11}=
\begin{pmatrix}
\ddots&\vdots&\vdots&\vdots&\vdots&\vdots&\Ddots\\
\hdots&0&0&0&0&0&\hdots\\
\hdots&0&0&0&0&0&\hdots\\
\hdots&0&0&0&0&0&\hdots\\
\hdots&0&0&0&1&0&\hdots\\
\hdots&0&0&0&0&0&\hdots\\
\Ddots&\vdots&\vdots&\vdots&\vdots&\vdots&\ddots
\end{pmatrix} .
\end{align}

Notice that
$$
D_{\omega,c}^n := \dist(\delta_{11}, \text{span}\{H_\omega^k \delta_{00}:k=0,1,2,\hdots,n\})
$$
simply describes the distance between the unit vector $\delta_{11}$ and the subspace obtained taking the closure of the span of the vectors $\delta_{00}, H_\omega \delta_{00}, H_\omega^2 \delta_{00}, \hdots, H_\omega^n \delta_{00}$.

In virtue of Corollary \ref{c-tool}, we obtain delocalization, if we can find $c>0$ for which \eqref{e-dist} happens with non-zero probability.


In the numerical experiment, we initially fix $c$ and fix one computer-generated realization of the random variable $\omega$ (with distribution in accordance to the hypotheses of Corollary \ref{c-tool}). We then calculate the distances $D_{\omega, c}^n$ for $n\in\{0, 1, 2, \hdots\}$.
In Subsection \ref{ss-comp} (below) we describe the numerical approach used to compute $D_{\omega,c}^n$.

Assuming that we know $D_{\omega,c}^n$ for $n=0, \hdots , 4500$, let us find a lower estimate for the limit
$$
 D_{\omega,c} : = \lim_{n\to\infty} D_{\omega,c}^n.
$$

\begin{figure}
 \centerline{
 \includegraphics[width=4.7in]{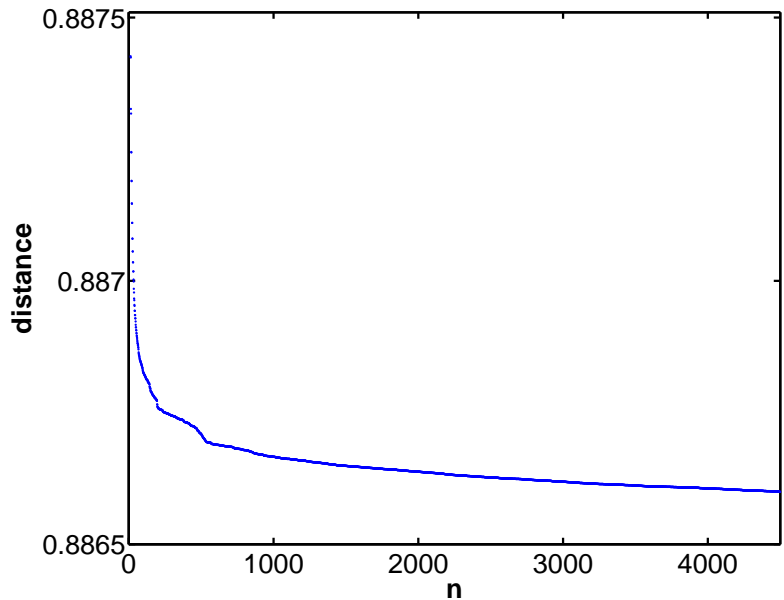}}
\caption{\label{graph-n} Typical trend for the distance $D_{\omega,c}^n$ as a function of $n$. Notice the fine vertical scale.}
\label{figure}
\end{figure}

Figure \ref{graph-n} shows typical trends for the distance $D_{\omega,c}^n$ as a function of $n$. As the first $n=200$ points do not contribute to the value of this limit and were generally rather irregular, we have omitted those points in what follows. (On the side, notice that the graph shows a decreasing function, as expected.)

While those results at this point looked fairly promising, they were not yet satisfactory. Most of all they do not provide a reliable estimate for the limit $D_{\omega,c}$. In order to obtain such an estimate for $D_{\omega,c}$, we re-scaled the horizontal axis in Figure \ref{graph-n} by a negative power $n^{-a}$ (power of the reciprocal, so that the horizontal axis is reverted) and approximated the resulting graph by a line.

The re-scaled graph is shown in Figure \ref{otherfigure}. Subsection \ref{ss-a} contains information about the choice of the re-scaling factor and explains why, for appropriately small disorders, the graph does not decay to zero, e.g.~logarithmically. The subtleties of choosing the re-scaling parameter $a$ are the reason why we do not expect delocalization with probability one in the Delocalization Conjecture \ref{t-mr}, but rather with non-zero probability. This decision is explained further in Subsection \ref{ss-a}

\begin{figure}
 \centerline{
 \includegraphics[width=4.7in]{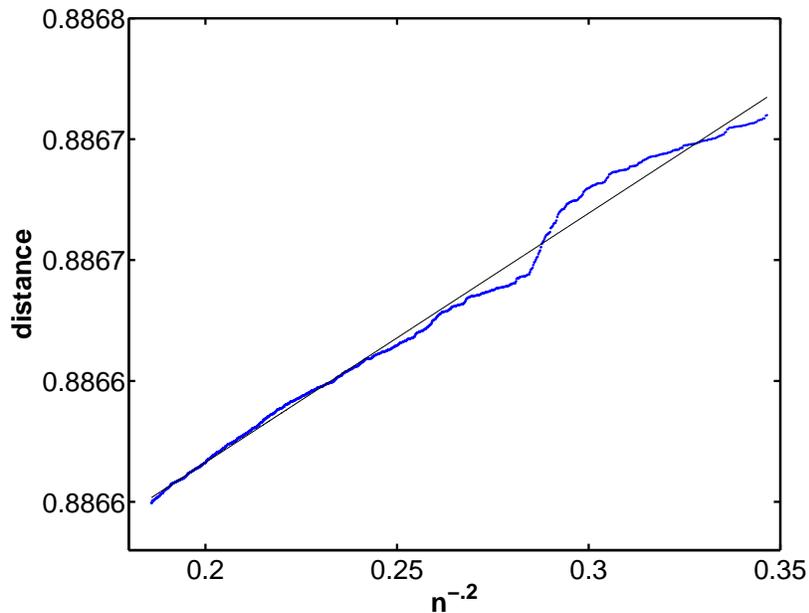}
 }
\caption{Re-scaling of the horizontal axis in Figure \ref{figure} using the best exponent $a=0.2$. The $y-$intercept of the approximating line is the estimate $y_{\omega,c}$ of the value for $D_{\omega,c}$.}\label{otherfigure}
\end{figure}

The value of $D_{\omega,c}$ is estimated by the $y-$intercept $y_{\omega,c}$ of the approximating line.

Since the re-scaled graphs in Figure \ref{otherfigure} were sometimes rather noisy (e.g.~a line through the steeper sections of the graph has a lower $y-$intercept), we decided to include a lower estimate $L_{\omega,c}$ for $y_{\omega,c}$ given by the minimum $y-$intercept of the lines passing through any two consecutive points.

Summarizing the last few steps, we have
\begin{align}\label{e-L}
D_{\omega,c} \approx y_{\omega,c} \ge L_{\omega,c}\,.
\end{align}

Finally, we repeat the experiment for many values of $c$ and many computer-generated realizations for the random variable $\omega$. 
Concerning the different realizations, throughout we took the minimum of $y_{\omega,c}$ and $L_{\omega,c}$ over all the different computer-generated realizations of $\omega$. Roughly the goal is to show that for some $c>0$, the limits $D_{\omega,c}$ are bounded away from zero for many realizations $\omega$.

\subsection{When to fix the realization $\omega$}
In the experiments described here, we had fixed $c$ and $\omega$ at the beginning. For fixed $c=0.1$, we also computed several cases for which we chose a different realization $\omega$ each time we applied he random operator. In other words, for $c=0.1$ we fix countably many realizations of $\omega$ each independently distributed and each in accordance with Corollary \ref{c-tool}. Let those realizations be denoted by $\omega_i$, $i \in \N$. Then we compute the distance between $\delta_{11}$ and the closure of the span of the vectors $\delta_{00}, \,H_{\omega_1}\delta_{00},\, H_{\omega_2}(H_{\omega_1}\delta_{00}), \,H_{\omega_3}(H_{\omega_2}H_{\omega_1}\delta_{00}),$ etc.

The results obtained from this setup agreed very well with the ones described in Section \ref{ss-results} below.

\subsection{Computing the distance $D_{\omega,c}^n$}\label{ss-comp}
For fixed $n$ and $\omega$ let us briefly explain the computational approach to obtain $D_{\omega,c}^n$. The main idea is to apply the Gram--Schmidt orthogonalization process in order to recursively compute $D_{\omega,c}^n$.

Take $m_0 = \delta_{00}$ and $D_{\omega,c}^0 = 1$ (since $\delta_{00}$ and $\delta_{11}$ are orthonormal).

In order to compute $D_{\omega,c}^{n+1}$, assume  we have an orthonormal basis $\{m_0, m_1, m_2, \hdots, m_n\}$ for the linear subspace  
$$
X_n:= 
\text{span}\{H_\omega^k \delta_{00}:k=0,1,2,\hdots,n\}
\text{ of }l^2(\Z^2).
$$

Let us find an orthonormal basis for $X_{n+1}$. According to the Gram--Schmidt orthogonalization process, we define $m_{n+1}$ to be the unit vector in the direction of
$$
H_\omega m_n - \sum_{l=0}^{n} < H_\omega m_n, m_{l}> m_{l} .
$$
The following proposition says that all but the two last terms in the sum are zero. 
We learned this fact and its proof from a conversation with M.~Hastings. This simplification reduces the required memory by the order $n$ (from $\mathcal{O}(n^3)$ to $\mathcal{O}(n^2)$).

\begin{prop}
The vector $H_\omega m_n$ is orthogonal to $m_l$ for all $l=0,1,2,\hdots,n-2$.
\end{prop}

Although this result seems to be well-known to the physics community, we include the short proof by mathematical induction on $n$.

\begin{proof}
Consider $n=2$. Assume that we have computed the orthonormal vectors $m_0$, $m_1$ and $m_2$ (via the Gram--Schmidt orthogonalization process).
Since the operator is self-adjoint, we have $< H_\omega m_2, m_{0}> = <m_2, H_\omega m_0>$. Since $m_1 = H_\omega m_0 - <H_\omega m_0,m_0>m_0$ and because $m_2$ is orthogonal to $m_0$ and $m_2$, we obtain
$$
< H_\omega m_2, m_{0}> =  \langle m_2, m_1 + <H_\omega m_0,m_0>m_0\rangle = 0.
$$

Assume that the statement of the proposition is true for some $n-1 \ge 2$. It remains to show that the statement is true for $n$. Assume that we have computed an orthonormal sequence $m_0, m_1, \hdots, m_{n}$.
For $l\le n-2$, it suffices to show that $< H_\omega m_{n}, m_{l}> =0$. By following the argument for the base case, we obtain $< H_\omega m_{n}, m_{l}> = < m_{n}, H_\omega m_{l}> = \langle  m_{n}, m_{l+1} + <H_\omega m_l,m_l>m_l\rangle$. The latter expression equals zero, because we assumed $l\le n-2$ and the orthogonality assumption on $m_0, m_1, \hdots, m_{n}$.
\end{proof}

According to the latter proposition we take $m_{n+1}$ to be the unit vector in the direction of
$$
m_{n+1} = \frac{\tilde m_{n+1}}{\|\tilde m_{n+1}\|_2}, \text{ where } \tilde m_{n+1} = H_\omega m_n - < H_\omega m_n, m_{n-1}> m_{n-1} - < H_\omega m_n, m_{n}> m_{n}.
$$


Now, the distance $D_{\omega,c}^{n+1}$ of the vector $\delta_{11}$ to the subspace $X_{n+1}$ equals the Euclidean norm
\begin{align}\label{e-D}
D_{\omega,c}^{n+1} = \|e_{n+1}\|_2 \text{ with }
e_{n+1} = \delta_{11} - P_{n+1} \delta_{11},
\end{align}
and where $P_{n+1}$ denotes the orthogonal projection from $l^2(\Z^2)$ onto $X_{n+1}$. 

A little more analysis allows us to simplify the latter expression. The following expression is closely related to the dimensionless scaling parameter that occurs in the so-called Thouless criterion.

\begin{prop}
We have $(D_{\omega,c}^{n+1})^2 = 1 - \sum_{l=2}^{n+1} (x_l)^2$ where $x_l$ denotes the $(1,1)-$entry of $m_l$.
\end{prop}

\begin{proof}
By the definition of the $m_n$'s, the vectors $e_{k}$ are recursively given by
$$
e_0 = \delta_{11}\text{ and } e_{n+1} = e_n - <e_n, m_{n+1}> m_{n+1}.
$$

We use this same recursive definition for $e_n$ in the inner product and the fact that the $m_n$ form an orthonormal sequence to obtain
$$
e_{n+1} = e_n - <(e_{n-1} - <e_{n-1}, m_{n}> m_{n}), m_{n+1}> m_{n+1} = e_n - <e_{n-1} , m_{n+1}> m_{n+1}.
$$
Repeated application of this argument yields that
$$
e_{n+1} = e_n - <e_{0} , m_{n+1}> m_{n+1} = e_n - <\delta_{11} , m_{n+1}> m_{n+1}.
$$

We proceed to replace $e_n$ by its recursive definition and so on, until we obtain
$$
e_{n+1}  = \delta_{11} - \sum_{l=2}^{n+1} <\delta_{11} , m_{l}> m_{l}.
$$

The proposition follows from equation \eqref{e-D} and the Pythagorean Theorem, since the $m_l$ form an orthonormal sequence and all of them are orthogonal to $e_{n+1}$. Also notice that $\|\delta_{11}\|_2=1$.
\end{proof}

\subsection{Choice of the re-scaling parameter}\label{ss-a}
For each fixed $c$ and $\omega$, the re-scaling exponent $a$ is chosen so that the re-scaled graph of the distance function (see Figure \ref{otherfigure}) satisfies the least square property; that is, the error when approximating the graph by a line is minimal. With this exponent we then find the corresponding linear approximation for the re-scaled distance function.

We include an extract of the table of best re-scaling exponents $a$ which satisfy the least square property for our data. As the values for $y_{\omega,c}$ were not very sensitively dependent on the precise value of $a$, we used a rather coarse mesh $a=0.05:0.05:0.85$ and refined using $a=0.01:0.01:0.05$, if the best re-scaling exponent was below $0.05$. Each entry in the table corresponds to a different realization of the random variable $\omega$.
\begin{align}\label{table}
\begin{array}{|c||c|c|c|c|c|c|c|c|c|c|c|c|c|c|}
\hline
c&.1&.15&.2&.3&.4&.5&.6&.7&.8&.9&1&1.2&1.3\\ \hline
&0.2&0.04&0.2&0.25&0.1&0.15&0.1&0.1&0.05&0.02&0.05&0.1&N/A\\ \hline
&0.1&0.2&0.25&0.1&0.04&0.1&0.1&0.05&0.05&0.02&N/A&N/A&0.05\\ \hline
&0.2&N/A&0.2&0.2&0.15&0.1&0.2&0.1&0.1&0.04&0.03&N/A&N/A\\ \hline
&0.05&0.1&0.05&0.15&0.15&0.1&0.05&0.1&0.05&0.03&0.05&N/A&N/A\\ \hline
\end{array}
\end{align}

The $N/A$ indicates that for this particular realization, even the re-scaling parameter $a=0.02$ yields a concave graph. For this realization, we do not obtain any information. For values of $c\gtrsim 1.2$ many realizations did not yield a reasonable best fit parameter $a$. No statement can be made for such disorders.

Since we can only investigate finitely many randomizations, and one of the realizations for a fairly small value of $c=0.15$ yielded an inconclusive result, we decided to conjecture delocalization with non-zero probability in the Delocalization Conjecture \ref{t-mr}, rather than almost surely.

The existence of a positive re-scaling factor implies that the graph in Figure \ref{figure} will not decay to zero, e.g.~logarithmically. Indeed, if we use a re-scaling factor smaller than the one in the table will result in a `globally concave' graph for the distances $D_{\omega,c}^n$. In this case, the $y-$intercept of the line lies below the value expected for $D_{\omega,c}^\infty$.

\section{Conclusions}\label{ss-results}
As mentioned in Section \ref{s-setup}, for a fixed $c$ we chose many realizations $\omega$. We took the minimum of the resulting quantities for $L_{\omega,c}$ and $y_{\omega,c}$ (the $y-$intercept of the approximating line and the minimum $y-$intercept of the lines passing through any two consecutive points, respectively).

\begin{figure}
 \centerline{
 \includegraphics[width=4.7in]{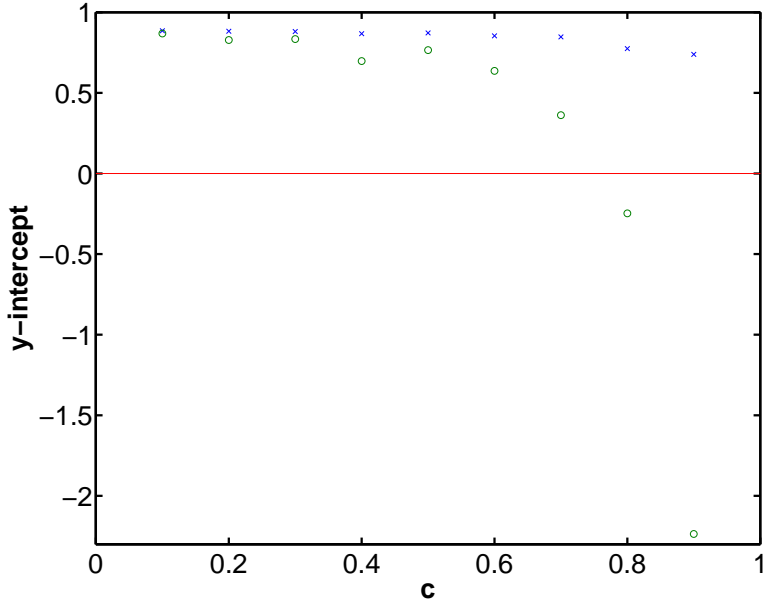}
 }
\caption{As a function of $c$ we show $y_{\omega,c}$ ($x$'s; larger values) and $L_{\omega,c}$ (circles; smaller values). Notice that $D_{\omega,c} \approx y_{\omega,c} \ge L_{\omega,c}>0$ for $c\lesssim0.7$, supporting our Delocalization Conjecture \ref{t-mr}.}\label{mainfigure}
\end{figure}

Figure \ref{mainfigure} shows $L_{\omega,c}$ and $y_{\omega,c}$ as a function of $c$. Being rather cautious, we say that a negative value for $y_{\omega,c}$ indicates that the orbit of $\delta_{00}$ may not span the whole space. Hence the final conclusion of this numerical experiment is precisely the Delocalization Conjecture \ref{t-mr}.

As it was explained in the remark following Corollary \ref{c-tool}, we cannot conclude localization even if $y_{\omega,c}<0$.
Therefore, the experiments do not imply localization for larger values of $c$.

\section{Further supporting the credibility of the method and the numerical experiments}\label{s-supp}
Apart from the usual tests (the program is running stably, checking all subroutines, many verifications for small $n$) , we have also tested the code and versions for other models: the free/unperturbed two dimensional Schr\"odinger operator and the one dimensional random Schr\"odinger operator. We briefly summarize the results, in order to provide verification for the correctness of method and code.

Further, we provide information of the energy distribution in terms of the distance from the origin of the evolution of the vector $\delta_{00}$, describing how a wave packet which was initially located at the origin changes as time progresses.


\subsection{Free discrete two dimensional Schr\"odinger operator}
When we apply the free discrete Schr\"odinger operator $H = H_{\bf 0}$ to the vector $\delta_{00}$, it immediately becomes clear that $H \delta_{00}$ as well as all vectors $H^n\delta_{00}$, $n\in\N\cup\{0\}$, are symmetric with respect to the origin. In dimension $d=2$, it is not hard to see that the distance between $\delta_{11}$ and the orbit of $\delta_{00}$ under $H$ is at least $\sqrt{3}/{2}\approx 0.8660$. Indeed, we have
$$
\dist(\delta_{11}, \clos\spa\{H^n \delta_{00} :n\in\N\cup\{0\}\}) > 
\min_x\dist(u_x, \delta_{11}) = \sqrt{3}/{2},
$$
where
$$
u_x = x\delta_{-1-1}+x\delta_{-11}+x\delta_{1-1}+x\delta_{11} = 
\begin{pmatrix}
\ddots&\vdots&\vdots&\vdots&\vdots&\vdots&\Ddots\\
\hdots&0&0&0&0&0&\hdots\\
\hdots&0&x&0&x&0&\hdots\\
\hdots&0&0&0&0&0&\hdots\\
\hdots&0&x&0&x&0&\hdots\\
\hdots&0&0&0&0&0&\hdots\\
\Ddots&\vdots&\vdots&\vdots&\vdots&\vdots&\ddots
\end{pmatrix}\,.
$$

In the experiments for the free discrete two dimensional Schr\"odinger operator we obtained a $y-$intercept of the approximating line approximately equals $0.8867$. The re-scaled graph of distances still had a very convex shape, so the actual distance as $n\to\infty$ would be bigger. In fact, we have extracted from Figure \ref{freelap} an upper estimate of $0.8868$ by zooming in. Therefore, the distance must lie in the interval $[0.8866, 0.8868]$.


\begin{figure}
 \centerline{
 \includegraphics[width=4.7in]{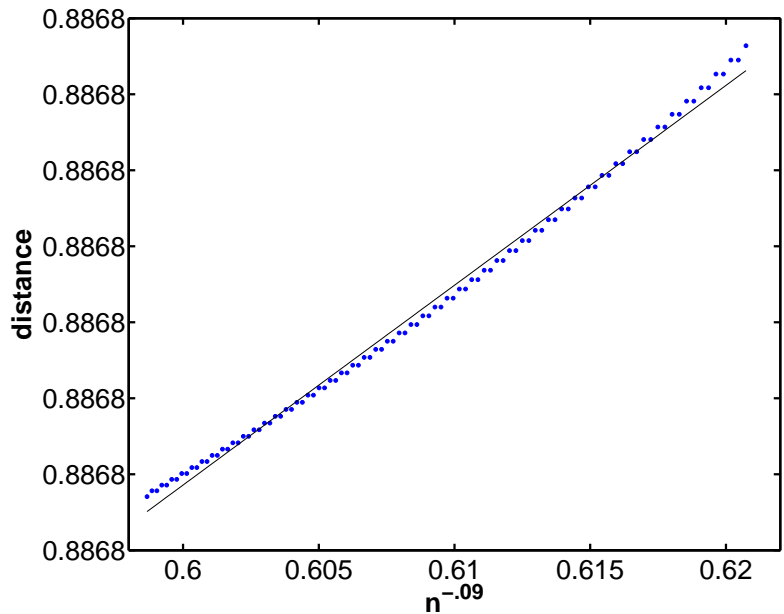}
 }
\caption{Convex shape for the distance from $\delta_{11}$ to the orbit of $\delta_{00}$ under the free Schr\"odinger operator in two dimensions. The approximating line has $y-$intercept $y_{\omega,c} \approx 0.8866$.}\label{freelap}
\end{figure}

\subsection{Verifying localization for the one dimensional random Schr\"odinger operator}
Consider the discrete random Schr\"odinger operator in one dimension, see e.g.~equations \eqref{Model} and \eqref{d-RandSchr} with $d=1$. For this operator, it is well known that localization occurs for random disorders of all strengths (in particular, for small values of $c$) and at all energies.

We have adopted and applied this computational approach for the discrete random Schr\"odinger operator in one dimension. Figure \ref{1Dfigure} shows a typical re-scaled graph of the distance
$$
D_{\omega,c}^{n} = \dist(\delta_{1}, \text{span}\{H_\omega^k \delta_{0}:k=0,1,2,\hdots,n\})
\text{ for } n=3000, 3001, 3002,  \hdots, 15000
$$
for the disorder $c = 0.05$.
With a re-scaling exponent of $a=0.09$, the graph of $D_{\omega,c}^{n}$ is still concave, so that the $y-$intercept of the approximating line is an upper estimate of $D_{\omega,c}$ the limit. Therefore we have
$$
D_{\omega,c}<y_{\omega,c} = -0.0543.
$$

While we know by the remark following Corollary \ref{c-tool} that this experiment does not allow us to conclude that there is localization, the result still provides support for the credibility of the method at hand as well as the numerical design.

\begin{figure}
 \centerline{
 \includegraphics[width=4.7in]{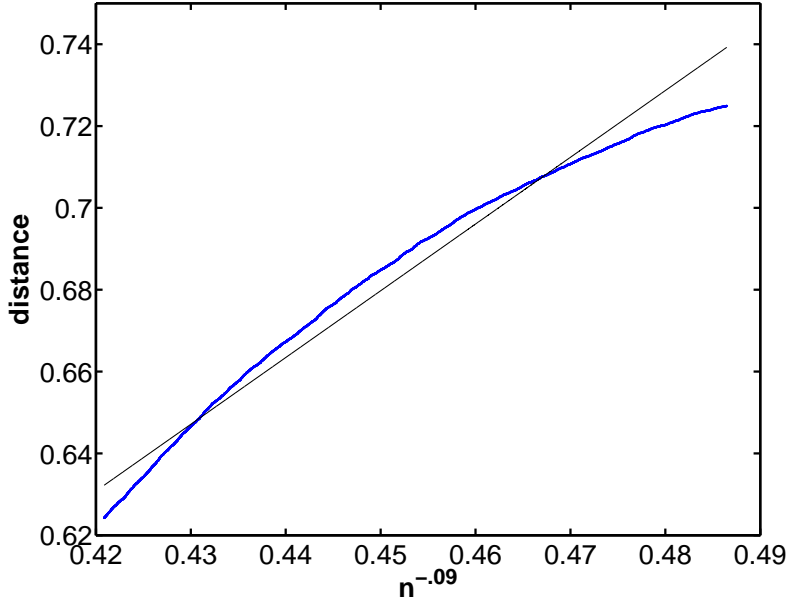}
 }
\caption{Discrete random Schr\"odinger operator in one dimension with disorder $c=0.05$. For $a=0.09$, we have $y_{\omega,c}\approx - 0.0543$ and a convex graph.}\label{1Dfigure}
\end{figure}

\subsection{Diffusion of energy for small values of $c$}\label{s-fresults}
We present the distribution of energies of a wave packet initially located at the origin as the random operator is repeatedly applied. By distribution of energies, we mean how much of the energy is located at which `distance' from the origin.

For example, in order to obtain how much energy of the vector $m_k$ (defined in Subsection \ref{ss-comp}) is at `distance' $2$ from the origin, we use the elements of $m_k$ which are located on the diamond for which the matrix
$$
\begin{pmatrix}
\ddots&\vdots&\vdots&\vdots&\vdots&\vdots&\Ddots\\
\hdots&4&3&2&3&4&\hdots\\
\hdots&3&2&1&2&3&\hdots\\
\hdots&2&1&0&1&2&\hdots\\
\hdots&3&2&1&2&3&\hdots\\
\hdots&4&3&2&3&4&\hdots\\
\Ddots&\vdots&\vdots&\vdots&\vdots&\vdots&\ddots
\end{pmatrix}
$$
has entries equal to $2$. The energy $E(2,k)$ of the vector $m_k$ at `distance' 2 from the origin is equal to the Euclidean norm over the elements in this diamond. In general, we have
\begin{align}\label{e-E}
E(l,k) = \sqrt{\sum_{|i|+|j|=l} (m_k)_{i,j}^2}
\end{align}
for the energy $E(l,k)$ of the vector $m_k$ at `distance' $l$ from the origin. Here $(m_k)_{i,j}$ refers to the $(i,j)-$entry of the $2\times 2-$ matrix $m_k$.

By small modifications of our programs, we have extracted the location of the energy the vector $\delta_{00}$ evolves under the random Hamiltonian for the values $c = .1,$ $c = 1$ and $c = 5$ of disorder, see Figure \ref{f-energy1}. In accordance with our Delocalization Conjecture \ref{t-mr}, the energy for small disorder is far away from the origin whereas it is concentrated close to the origin for large disorder.

\begin{figure}
 \centerline{
  \includegraphics[width=4.7in]{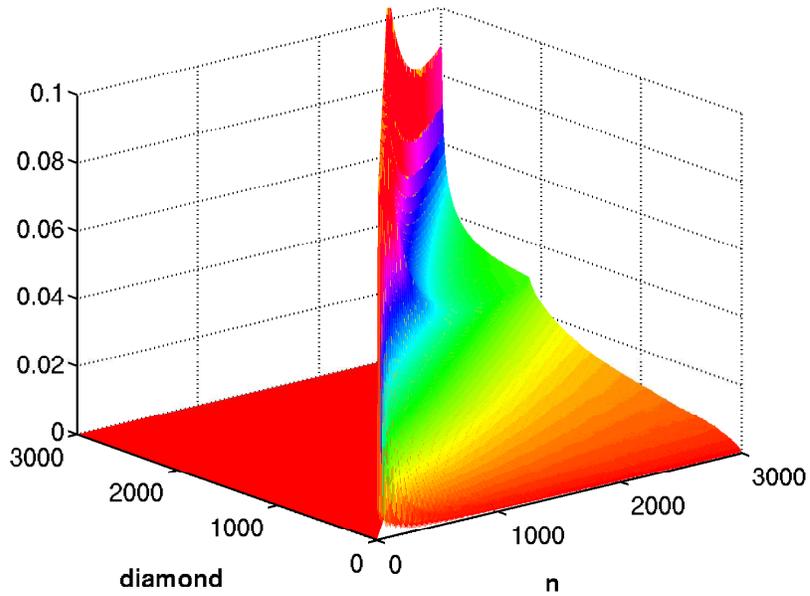}
 }
\caption{For $c=0.05$ we show $E(l,n)$, i.e.~the evolution of the energy distribution of $H_{\omega,c}^n \delta_{00}$ for the diamonds at distance $l$ from the origin. Notice that the energy travels far out from the origin (the diagonal is the farthest possible).}\label{f-energy2}
\end{figure}

\begin{figure}
 \centerline{
 \includegraphics[width=4.7in]{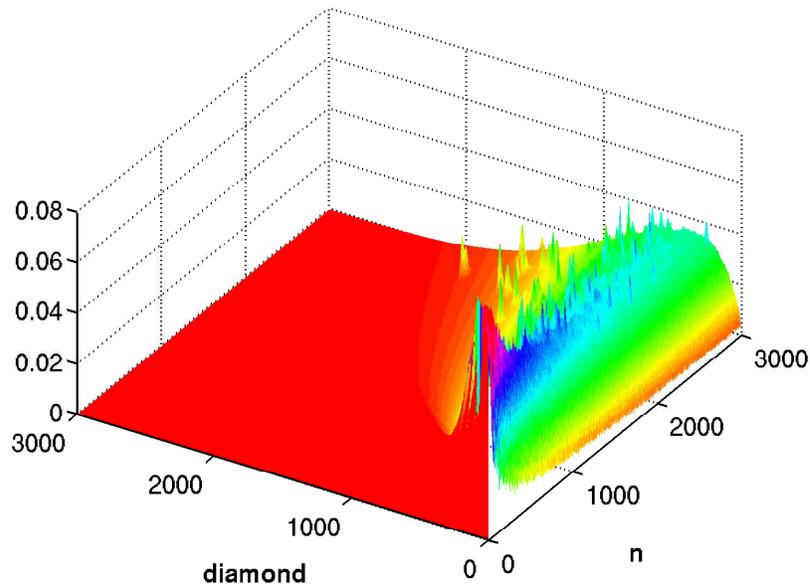}
}
\caption{The analog of Figure \ref{f-energy2} for $c=0.6$ (but with different color scale). The energy remains much closer to the origin.}\label{f-energy3}
\end{figure}

Figures \ref{f-energy2} and \ref{f-energy3} shows the energy distribution of $H_\omega^n\delta_{00}$ for $n=2999$ for values of $c$ ranging from $c=.1$ to $c=1$. Again, the fact that the energy for small disorder is far away from the origin whereas it shifts much closer to the origin as the disorder increases, supports Delocalization Conjecture \ref{t-mr}.

\begin{figure}
 \centerline{
 \includegraphics[width=4.9in]{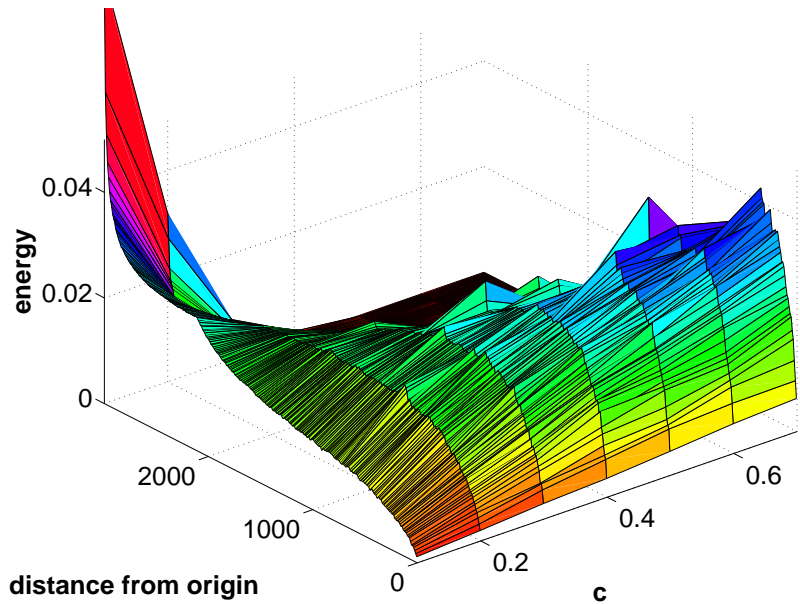}
 }
\caption{The figure shows $E(l,2999)$ (the Euclidean norm on the diamonds of $m_{2999}$ given by equation \eqref{e-E}) as a function of the `distance' $l$ from the origin for several values of disorder $c$.}\label{f-energy1}
\end{figure}

Both figures are the averages obtained from two realizations for each value of $c$. And again, many repetitions of these experiments for smaller values of $n$ were carried out, and the figures shown represent the behavior obtained in all repetitions.

\subsection{Precision}
The results are not a phenomenon of numerical errors (e.g.~round off errors that sum up over time). Indeed, we compared our results with those of a double precision computation. The results agreed very well.

%


\section{On computing and memory requirements}\label{s-final}
The implementation uses memory rather efficiently, so that the numerical experiments were mainly limited by the length of the computation. On the rather small machines available to us, it took 8 1/2 hours to complete one realization for one value of $c$. Since we need to include many realizations of the random variable and many values of $c$, it took even several units a considerable time to finish all the computations.

In order to compute $D_\omega^n$ described in Section \ref{s-setup}, our code requires order $n^2$ (i.e.~$\mathcal{O}(n^2)$) memory. Indeed, in order to carry out the Gram--Schmidt orthogonalization process described in Subsection \ref{ss-comp}, we must store matrices of size $\mathcal{O}(n)\times \mathcal{O}(n)$.
The corresponding code for the $d-$dimensional discrete random Schr\"odinger operator will require memory size of order $\mathcal{O}(n^{d})$.

The random Schr\"odinger operator on the (dyadic) tree uses memory of $\mathcal{O}(2^n)$. With the resources available to us, memory restrictions would only allow us to compute up to $n\approx 27$ for the tree. In this case, we cannot produce sufficient data to support the fact that the discrete random Schr\"odinger operator on the tree does indeed exhibit delocalization.

\providecommand{\bysame}{\leavevmode\hbox to3em{\hrulefill}\thinspace}
\providecommand{\MR}{\relax\ifhmode\unskip\space\fi MR }
\providecommand{\MRhref}[2]{%
  \href{http://www.ams.org/mathscinet-getitem?mr=#1}{#2}
}
\providecommand{\href}[2]{#2}


\begin{thebibliography}{10}
\bibitem[Abakumov--Liaw--Poltoratski{\u\i} 2012]{AbaLiawPolt}
E.~Abakumov, C.~Liaw, A.~Poltoratski{\u\i}, \emph{Cyclic vectors for rank-one perturbations and Anderson-type Hamiltonians}. Submitted, also see {\tt arXiv:1111.3095}.


 \bibitem[Anderson 1958]{And1958}
 P.~W.~Anderson, \emph{Absence of Diffusion in Certain Random Lattices}, Phys. Rev., \textbf{109}
   (1958), 1492--1505.

\bibitem[Aizenman--Molchanov 1993]{AizMol1993}
M.~Aizenman, S.~Molchanov, \emph{{Localization at large disorder and at extreme energies: An elementary derivation}}, Comm.~Math.~Phys.~\textbf{157} (1993), no.~2, 245--278.

\bibitem[Birman--Solomjak 1986]{birst}
M.~S.~Birman, M.~Z.~Solomjak, \emph{{Spectral theory of self-adjoint operators
  in Hilbert space}}, 1986.

\bibitem[Carmona--Lacroix 1990]{CL}
R.~Carmona, J.~Lacroix, \emph{{Spectral theory of random Schr\"odinger operators}}, Birkh\"auser, 1990.

\bibitem[Cycon--Froese--Kirsh--Simon 1987]{CFKS}
H.~Cycon, R.~Froese, W.~Kirsh, B.~Simon, \emph{{Topics in the Theory of Schr\"odinger Operators}}, Springer Verlag, 1987.

\bibitem[del Rio--Jitomirskaya--Last--Simon 1986]{SIMSULE}
R.~del Rio, S.~Jitomirskaya, Y.~Last, B.~Simon, \emph{Operators with
  singular continuous spectrum. {IV}. {H}ausdorff dimensions, rank-one
  perturbations, and localization}, J. Anal. Math. \textbf{69} (1996),
  153--200. \MR{1428099 (97m:47002)}

\bibitem[Figotin--Pastur 1991]{FP}
A.~Figotin, L.~Pastur, \emph{{Spectral properties of disordered systems in the one-body approximation}}, Springer Verlag, 1991.

\bibitem[Fr\"ohlich--Spencer 1983]{FS}
J.~Fr\"ohlich, T.~Spencer, \emph{{Absence of Diffusion in the tight binding model for large disorder of low energy}}, Commun.~Math.~Phys.~\textbf{88} (1983), 151--184.

\bibitem[Germinet--Klein--Schenker 2007]{Germ}
F.~Germinet, A.~Klein, J.~H.~Schenker, \emph{Dynamical
  delocalization in random {L}andau {H}amiltonians}, Ann. of Math. (2)
  \textbf{166} (2007), no.~1, 215--244. \MR{2342695 (2008k:82060)}

\bibitem[Ghribi--Hislop--Klopp 2007]{Klo2}
F.~Ghribi, P.~D.~Hislop, F.~Klopp, \emph{Localization for {S}chr\"odinger
  operators with random vector potentials},  \textbf{447} (2007), 123--138.
  \MR{2423576 (2009d:82067)}

\bibitem[Jak{\v{s}}i{\'c}--Last 2000]{JakLast2000}
V.~Jak{\v{s}}i{\'c}, Y.~Last, \emph{Spectral structure of {A}nderson
  type {H}amiltonians}, Invent. Math. \textbf{141} (2000), no.~3, 561--577.
  \MR{1779620 (2001g:47069)}

\bibitem[Jak{\v{s}}i{\'c}--Last 2006]{JakLast2006}
\bysame, \emph{Simplicity of singular spectrum in {A}nderson-type
  {H}amiltonians}, Duke Math. J. \textbf{133} (2006), no.~1, 185--204.
  \MR{2219273 (2007g:47062)}

\bibitem[Kato 1980]{katobook}
T.~Kato, \emph{Perturbation theory for linear operators}, Classics in
  Mathematics, Springer Verlag, Berlin, 1995, Reprint of the 1980 edition.
  \MR{1335452 (96a:47025)}

 \bibitem[Last 2007]{ExSpec}
 Y.~Last, \emph{Exotic Spectra: A Review of Barry Simon's Central Contributions}, Proceedings of Symposia in Pure Mathematics \textbf{76.2.} (2007) 697--712.

\bibitem[Simon 1994A]{Sim1994}
B.~Simon, \emph{Cyclic vectors in the {A}nderson model}, Rev. Math. Phys.
  \textbf{6} (1994), no.~5A, 1183--1185, Special issue dedicated to Elliott H.
  Lieb. \MR{1301372 (95i:82058)}

\bibitem[Simon 1994B]{SIMREV}
\bysame, \emph{{Spectral analysis of rank-one perturbations and applications}},
  Mathematical Quantum Theory I: Field Theory and Many-Body Theory (1994).

\bibitem[Simon--Wolff 1986]{SIMWOL}
B.~Simon and T.~Wolff, \emph{Singular continuous spectrum under rank-one
  perturbations and localization for random {H}amiltonians}, Comm.~Pure Appl.~Math., \textbf{39} (1986), no.~1, 75--90. \MR{820340 (87k:47032)}
\end{thebibliography}
\end{document}